\documentclass[11pt]{article}

\usepackage[T1]{fontenc}
\usepackage{textcomp}
\usepackage{palatino}
\usepackage{mathpazo}
\usepackage{stmaryrd}

\usepackage{hyperref}
\hypersetup{pdfpagemode=UseNone}

\usepackage{fullpage}
\usepackage{amsfonts}
\usepackage{amssymb}
\usepackage{amsmath}
\usepackage{latexsym}
\usepackage{amsthm}
\usepackage{eepic}
\usepackage{sectsty}
\usepackage[usenames]{color}

\newtheorem{theorem}{Theorem}
\newtheorem{lemma}[theorem]{Lemma}

\theoremstyle{definition}

\newtheorem{claim}[theorem]{Claim}

\newtheorem{fact}[theorem]{Fact}

\newenvironment{proofof}[1]{\noindent{\bf Proof of #1:}}{\qed\\}

\newcommand{\tinyspace}{\mspace{1mu}}

\newcommand{\norm}[1]{\left\lVert\tinyspace#1\tinyspace\right\rVert}

\newcommand{\defeq}{\stackrel{\smash{\text{\tiny def}}}{=}}
\newcommand{\tr}{\operatorname{Tr}}

\def\({\left(}
\def\){\right)}
\def\I{\mathbb{1}}

\def\complex{\mathbb{C}}

\def\O{\mathcal{O}}

\def\ve{{\varepsilon}}
\def\poly{\mathrm{poly}}
\def\polylog{\mathrm{polylog}}
\def\opt{\mathrm{opt}}
\def\thr{\mathsf{thr}}
\def\Diag{\mathrm{Diag}}

\newcommand{\suppress}[1]{}

\title{A Parallel Approximation Algorithm for Positive Semidefinite Programming}

\author{
Rahul Jain\thanks{
Centre for Quantum Technologies and Department of Computer Science, 
National University of Singapore,  Block S15, 3 Science Drive~2,
Singapore 11754.
Email: \texttt{rahul@comp.nus.edu.sg}.
} \\
National U.\ Singapore
\and
Penghui Yao\thanks{
Centre for Quantum Technologies and Department of Computer Science, 
National University of Singapore,  Block S15, 3 Science Drive~2,
Singapore 11754.
Email: \texttt{pyao@nus.edu.sg}.
} \\
National U.\ Singapore
}

\begin{document}

\maketitle 

\abstract{Positive semidefinite programs are an important subclass of semidefinite programs in which all matrices involved in the specification of the problem are positive semidefinite and all scalars involved are non-negative. We present a parallel algorithm, which given an instance of a positive semidefinite program of size $N$ and an approximation factor $\ve > 0 $, runs in (parallel) time $\poly(\frac{1}{\ve}) \cdot \polylog(N)$, using $\poly(N)$ processors, and outputs a value which is within multiplicative factor of $(1+\ve)$ to the optimal. Our result generalizes analogous result  of Luby and Nisan~\cite{LubyN93} for positive linear programs and our algorithm is inspired by the algorithm of~\cite{LubyN93}.}

\thispagestyle{empty} 

\newpage 

\setcounter{page}{1}

\section{Introduction}
Fast parallel algorithms for approximating optimum solutions to different subclasses of semidefinite programs have been studied in several recent works (e.g.~\cite{AroraHK05, AroraK07, Kale07, JainW09, JainUW09, JainJUW10}) leading to many interesting applications including the celebrated result $\mathrm{QIP=PSPACE}$ \cite{JainJUW10}. However for each of the algorithms used for example in~\cite{JainW09, JainUW09, JainJUW10}, in order to produce a $(1+\ve)$ approximation of the optimal value for a given semidefinite program of size $N$, in the corresponding subclass that they considered, the (parallel) running time was $\polylog(N) \cdot \poly(\kappa) \cdot \poly(\frac{1}{\ve})$, where $\kappa$ was a 'width' parameter that depended on the input semidefinite program (and was defined differently for each of the algorithms). For the specific instances of the semidefinite programs arising out of the applications considered in~\cite{JainW09, JainUW09, JainJUW10}, it was separately argued that the corresponding 'width parameter' $\kappa$ is at most $\polylog(N)$ and therefore the running time remained  $\polylog(N)$ (for constant $\ve$). It was therefore desirable to remove the polynomial dependence on the 'width' parameter and obtain a truly $\polylog$ running time algorithm, for a reasonably large subclass of semidefinite programs. 

In this work we consider the class of positive semidefinite programs. A positive semidefinite program can be expressed in the following standard form (we use symbols $\geq, \leq$ to also represent L\"owner order).
\begin{center}
  \begin{minipage}{2in}\vspace{-10mm}
    \centerline{\underline{Primal problem ${P}$}}\vspace{-4mm}
    \begin{align*}
      \text{minimize:}\quad & \tr {C} X\\
      \text{subject to:}\quad & \forall i\in  [m] : \tr {A}_i X \geq {b}_i,\\
      & X \geq 0.
    \end{align*}
  \end{minipage}
  \hspace*{25mm}
  \begin{minipage}{2in}
    \centerline{\underline{Dual problem ${D}$}}\vspace{-4mm}
    \begin{align*}
      \text{maximize:}\quad & \sum_{i=1}^m {b}_i y_i \\
      \text{subject to:}\quad & \sum_{i=1}^m y_i \cdot {A}_i  \leq {C},\\
      & \forall i \in [m] : y_i \geq 0.
    \end{align*}
  \end{minipage}
\end{center}
Here $C, A_1, \ldots, A_m$ are $n \times n$ positive semidefinite matrices and $b_1, \ldots, b_m$ are non-negative reals (in a general semidefinite program $C, A_1, \ldots, A_m$ are Hermitian and $b_1, \ldots, b_m$ are reals).  Let us assume that the conditions for strong duality are satisfied and the optimum value for ${P}$, denoted $\opt({P})$, equals the optimum value for ${D}$, denoted $\opt({D})$. We present an algorithm, which given as input, $(C, A_1, \ldots, A_m, b_1, \ldots, b_m)$, and an error parameter $\ve > 0$, outputs a $(1 + \ve)$ approximation to the optimum value of the program, and has running time $\polylog(n) \cdot \polylog(m) \cdot \poly(\frac{1}{\ve})$. As can be noted, there is no polynomial dependence on any 'width' parameter on the running time of our algorithm. The classes of semidefinite programs used in~\cite{JainW09, JainUW09} are a subclass of positive semidefinite programs and hence our algorithm can directly be applied to the programs in them without needing any other argument about the 'width' being $\polylog$ in the size of the program (to obtain an algorithm running in $\polylog$ time).

Our algorithm is inspired by the algorithm used by Luby and Nisan~\cite{LubyN93} to solve positive linear programs. Positive linear programs can be considered as a special case of positive semidefinite programs in which the matrices used in the description of the program are all pairwise commuting.  Our algorithm (and the algorithm in~\cite{LubyN93}) is based on the 'multiplicative weights update' (MWU) method. This is a powerful technique for 'experts learning' and finds its origins in various fields including learning theory, game theory, and optimization. The algorithms used in~\cite{AroraHK05, AroraK07, Kale07, JainW09, JainUW09, JainJUW10} are based on its matrix variant the 'matrix multiplicative weights update' method. The algorithm of Luby and Nisan~\cite{LubyN93} proceeds in phases, where in each phase the large eigenvalues of $\sum_{i=1}^m y^t_i A_i$ ($y^t_i$s represent the candidate dual variables at time $t$) are sought to be brought below a threshold determined for that phase. The primal variable at time step $t$ is chosen to be the projection onto the large eigenvalues (above the threshold) eigenspace of $\sum_{i=1}^m y^t_i A_i$. Using the sum of the primal variables generated so far, the dual variables are updated using the MWU method. A suitable scaling parameter $\lambda_t$ is chosen during this update, which is small enough so that the good properties needed in the analysis of MWU are preserved and at the same time is large enough so that there is reasonable progress in bringing down the large eigenvalues. 

Due to the non-commutative nature of the matrices involved in our case, our algorithm primarily deviates from that of~\cite{LubyN93} in how the threshold is determined inside each phase. The problem that is faced is roughly as follows. Since $A_i$'s could be non-commuting, when $y^t_i$s are scaled down, the sum of the large eigenvalues of $\sum_{i=1}^m y^t_i A_i$ may not come down and this scaling may just move the large eigenvalues eigenspace. Therefore a suitable extra condition needs to be ensured while choosing the threshold.  Due to this, our analysis also primarily deviates from~\cite{LubyN93} in bounding the number of time steps required in any phase and is significantly more involved. The analysis requires us to study the relationship between the large eigenvalues eigenspaces before and after scaling (say $W_1$ and $W_2$). For this purpose we consider the decomposition of the underlying space into one and two-dimensional subspaces which are invariant under the actions of both $\Pi_1$ and $\Pi_2$ (projections onto $W_1$ and $W_2$ respectively) and this helps the analysis significantly. Such decomposition has been quite useful in earlier works as well for example in quantum walk~\cite{Szegedy04,Regev06,AmbainisCRSZ07} and quantum complexity theory~\cite{MarriottW05,DanielWZ09}.

We present the algorithm in the next section and its analysis, both optimality and the running time, in the subsequent section. Due to space constraints we move some proofs to the Appendix.

\suppress{Lemma \ref{lem:projectordecomposition} is an important result
related to quantum walk and quantum complexity theory
\cite{sz04,reg06,crsz07,nwz09,rei09} (This is what I wanna emphasize
in the introduction)}

\section{Algorithm}
Given the positive semidefinite program $({P},{D})$ as above, we first show in Appendix~\ref{sec:transform} that without loss of generality $(P, D)$ can be in the following special form.

\begin{center}
  \begin{minipage}{2in}\vspace{-10mm}
    \centerline{\underline{Special form Primal problem $P$}}\vspace{-4mm}
    \begin{align*}
      \text{minimize:}\quad & \tr X\\
      \text{subject to:}\quad & \forall i\in  [m] : \tr A_i X \geq 1,\\
      & X \geq 0.
    \end{align*}
  \end{minipage}
  \hspace*{25mm}
  \begin{minipage}{2in}
    \centerline{\underline{Special form Dual problem $D$}}\vspace{-4mm}
    \begin{align*}
      \text{maximize:}\quad & \sum_{i=1}^m y_i \\
      \text{subject to:}\quad & \sum_{i=1}^m y_i \cdot A_i  \leq \I,\\
      & \forall i \in [m] : y_i \geq 0.
    \end{align*}
  \end{minipage}
\end{center}
Here $A_1, \ldots, A_m$ are $n \times n$ positive semidefinite matrices and $\I$ represents the identity matrix. Furthermore, for all $i$, norm of $A_i$, denoted $\norm{A_i}$, is at most  $1$ and the minimum non-zero eigenvalue of $A_i$ is at least $\frac{1}{\gamma}$ where $\gamma = \frac{m^2}{\ve^2}$.

In order to compactly describe the algorithm, and also the subsequent analysis, we introduce some notation. Let $Y = \Diag(y_1, \ldots , y_m)$ ($m \times m$ diagonal matrix
with $Y(i,i) = y_i$ for $i \in [m]$). Let $\Phi$ be the map (from $n \times n$ positive semidefinite matrices to $m \times m$  positive semidefinite diagonal matrices) defined
by $\Phi(X) = \Diag(\tr A_1 X, \ldots ,\tr A_m X )$. Then its
adjoint map $\Phi^*$ acts as $\Phi^*(Y) = \sum_{i=1}^m Y(i,i) \cdot
A_i$ (for all diagonal matrices $Y \geq 0$).  We let $\I$ represent the identity matrix (in the appropriate dimensions clear from the context). For Hermitian matrix $B$ and real number $l$, let $N_l(B)$ represent the sum of eigenvalues of $B$ which are at least $l$. 
The algorithm is mentioned in Figure~\ref{fig:alg}. 

\begin{figure}[!ht]

\noindent\hrulefill

{
\small

\noindent {\bf Input :} Positive semidefinite matrices $A_1, \ldots, A_m$ and error parameter $\ve >0$.

\medskip

\noindent {\bf Output :} $X^*$ feasible for $P$ and $Y^*$ feasible for $D$. 

\begin{enumerate}

\item Let $\ve_0=\frac{\ve^2}{\ln^2 n}$, $t=0, X_0=0$. Let $k_s$ be the
smallest positive number such that
$(1+\ve_0)^{k_s}\leq\|\Phi^*(\I)\|<(1+\ve_0)^{k_s+1}$. Let
$k=k_s$.

\item Let $Y_t  = \exp(-\Phi(X_{t}))$.

\item If $\tr Y_t>\frac{1}{m^{1/\ve}}$, do

\begin{enumerate}

\item If $\|\Phi^*(Y_t)\|<(1+\ve_0)^k$, then set
$k\leftarrow k-1$ and repeat this step.

\item Set $\thr'=k$.

\item If
$$N_{(1+\ve_0)^{\thr'-1}}(\Phi^*(Y_t))\geq(1+\frac{2}{5}\ve)N_{(1+\ve_0)^{\thr'}}(\Phi^*(Y_t)).$$
then $\thr'\leftarrow  \thr'-1$ and repeat this step. Else set $\thr = \thr'$.

\item Let $\Pi_t$ be the projector on the eigenspace of
$\Phi^*(Y_t)$ with eigenvalues at least $(1+\ve_0)^{\thr}$. For
$\lambda>0$, let $P^{\geq}_{\lambda}$ be the projection onto
eigenspace of $\Phi(\lambda\Pi_t)$ with eigenvalues at least
$2\sqrt{\ve}$. Let $P^{\leq}_{\lambda}$ be the projection onto
eigenspace of $\Phi(\lambda\Pi_t)$ with eigenvalues at most
$2\sqrt{\ve}$. Find $\lambda_t$  such that

\begin{itemize}

\item[1.] $\tr (P^\geq_{\lambda_t} Y_t P^\geq_{\lambda_t}) \Phi(\Pi_t) \geq \sqrt{\ve} \tr  Y_t  \Phi(\Pi_t)$ and,

\item[2.] $\tr (P^\leq_{\lambda_t} Y_t P^\leq_{\lambda_t}) \Phi(\Pi_t) \geq (1 - \sqrt{\ve}) \tr  Y_t  \Phi(\Pi_t)$ as follows.

\end{itemize}

\begin{enumerate}
\item Sort $\{\tr A_i \Pi_t\}_{i=1}^m$ in non-increasing order. Suppose $\tr A_{j_1}\Pi_t\geq\tr A_{j_2}\Pi_t\geq\cdots\geq\tr A_{j_m}\Pi_t.$
\item Find index $r\in [m]$ satisfying
$$\sum_{k=1}^r y_{j_k} \tr A_{j_k}\Pi_t\geq\sqrt{\ve}\sum_{k=1}^m y_{j_k}\tr A_{j_k}\Pi_t, \mbox{ and } $$ 
$$\sum_{k=r}^m y_{j_k} \tr A_{j_k}\Pi_t\geq(1-\sqrt{\ve})\sum_{k=1}^m y_{j_k} \tr A_{j_k}\Pi_t.$$
\item Let $\lambda_t= \frac{2\sqrt{\ve}}{\tr A_{j_r}\Pi_t}$. 
\end{enumerate}

\item Let $X_{t+1} = X_t + \lambda_t \Pi_t$. Set $t \leftarrow t+1$ and go to Step 2.

\end{enumerate}

\item Let $t_f = t$, $k_f = k$. Let $\alpha$ be the minimum eigenvalue of $\Phi(X_{t_f})$. Output $X^* = X_{t_f}/\alpha$.
\item Let $t'$ be such that $\tr Y_{t'} / \norm{\Phi^*(Y_{t'})}$ is the maximum among all time steps. Output $Y^* = Y_{t'}/\norm{\Phi^*(Y_{t'})}$.

\end{enumerate}

}
\noindent\hrulefill 

\caption{Algorithm} \label{fig:alg}

\end{figure}

%\begin{figure}[!h]
%\label{flambda}
%\noindent\hrulefill
%\begin{enumerate}
%\item Sort $\{\tr A_i \Pi_t\}_{i=1}^m$ in non-increasing order. Suppose $\tr A_{j_1}\Pi_t\geq\tr A_{j_2}\Pi_t\geq\cdots\geq\tr A_{j_m}\Pi_t.$
%\item Find index $r$ satisfying
%$$\sum_{k=1}^r y_{j_k} \tr A_{j_k}\Pi_t\geq\sqrt{\ve}\sum_{k=1}^m y_{j_k}\tr A_{j_k}\Pi_t,$$ 
%$$\sum_{k=r}^m y_{j_k} \tr A_{j_k}\Pi_t\geq(1-\sqrt{\ve})\sum_{k=1}^m y_{j_k} \tr A_{j_k}\Pi_t.$$
%\item Set $\lambda_t= \frac{\sqrt{\ve}}{\tr A_{j_r}\Pi_t}$. 
%\end{enumerate}
%
%\noindent\hrulefill \caption{Computing $\lambda_t$} 
%\end{figure}

\section{Analysis}
For all of this section, let $\ve_1 = \frac{3\ve}{\ln n}$. In the following we assume that $n$ is sufficiently large and $\ve$ is sufficiently small. 

\subsection{Optimality}
In this section we present the analysis assuming that all the operations performed by the algorithm are perfect. We claim, without going into further details, that similar analysis can be performed while taking into account the accuracy loss due to the actual operations of the algorithm in the limited running time.

We start with following claims.
\begin{claim}\label{claim:lambda}
For all $t \leq t_f$, $\lambda_t$ satisfies the conditions $1.$ and $2.$ in Step (3d) in the Algorithm.
\end{claim}
\begin{proof}
Easily verified.
\end{proof}
\begin{claim}\label{claim:alpha}
$\alpha > 0$.
\end{claim}
\begin{proof}
Follows since $\frac{1}{m^{1/\ve}} \geq \tr Y_{t_f} = \tr \exp(- \Phi(X_{t_f})) >  \exp (- \alpha) \enspace .$
\end{proof}
Following lemma shows that for any time $t$, $\|\Phi^*(Y_t)\|$ is not much larger than $(1+\ve_0)^{\thr}$.
\begin{lemma}
\label{lem:normphi}
For all $t \leq t_f$, $\|\Phi^*(Y_t)\|\leq (1+\ve_0)^{\thr}(1+\ve_1).$ 
\end{lemma}
\begin{proof}
Fix any $t \leq t_f$. As $\tr(\Phi^*(Y_t))\leq
nN_{(1+\ve_0)^k}(\Phi^*(Y_t))$, the loop at Step 3(c) runs at most
$\frac{\ln n}{\ln(1+\frac{2\ve}{5})}$ times. Hence
\begin{align*}
\|\Phi^*(Y_t)\| & \leq(1+\ve_0)^{k+1} \leq(1+\ve_0)^{\thr}(1+\ve_0)^{\frac{\ln
n}{\ln(1+\frac{2\epsilon }{5})}+1} \\
& <(1+\ve_0)^{\thr}(1+\frac{3\ve}{\ln
n})=(1+\ve_0)^{\thr}(1+\ve_1).
\end{align*}
\end{proof}
Following lemma shows that as $t$ increases, there is a reduction in the trace of the dual variable in terms of the trace of the primal variable.
\begin{lemma}\label{lem:thr}
For all $t \leq t_f$ we have, $\tr Y_{t+1} \leq \tr Y_t  - \lambda_t
\cdot (1 - 4\sqrt{\ve}) \cdot \norm{\Phi^*(Y_t)} \cdot (\tr \Pi_t)
\enspace .$
\end{lemma}
\begin{proof}
Fix any $t\leq t_f$. Let $B = P^\leq_{\lambda_t} \Phi(\lambda_t \Pi_t)
P^\leq_{\lambda_t}$. Note that $B \leq \Phi(\lambda_t \Pi_t)$ and
also $B \leq 2\sqrt{\ve} \I$. Second last inequality below follows from Lemma~\ref{lem:normphi} which shows that all eigenvalues of $\Pi_t
\Phi^*(Y_t) \Pi_t$ are at least $(1-\ve_1) \norm{\Phi^*(Y_t)}$.
\begin{align*}
\tr Y_{t+1} & = \tr \exp (-\Phi (X_t)   - \Phi(\lambda_t \Pi_t)) \\
& \leq \tr \exp (-\Phi (X_t) - B) \quad ( \text{since for } A_1 \geq A_2, \tr \exp(A_1) \geq \tr \exp(A_2)) \\
& = \tr \exp(-\Phi (X_t)) \exp(-B) \\ % \quad ( \text{both $B$} )
& \leq \tr \exp(-\Phi (X_t)) ( \I - (1-2\sqrt{\ve}) B) \quad ( \text{since for } A \leq \I : \exp(-\delta A) \leq \I - \delta(1-\delta)A)\\
& = \tr Y_t - (1-2\sqrt{\ve}) \tr Y_t  B \\
& \leq \tr Y_t - (1-\sqrt{\ve}) (1-2\sqrt{\ve}) \tr Y_t  \Phi(\lambda_t \Pi_t) \quad (\text{from step 3(d) part 1.}) \\
& =  \tr Y_t - (1-\sqrt{\ve}) (1-2\sqrt{\ve}) \tr \Phi^*(Y_t)  \lambda_t \Pi_t \\
& \leq \tr Y_t - (1-\ve_1)(1-\sqrt{\ve}) (1-2\sqrt{\ve}) \lambda_t
\norm{\Phi^*(Y_t)} (\tr \Pi_t) \enspace \\
& \leq \tr Y_t - (1-4\sqrt{\ve}) \lambda_t \norm{\Phi^*(Y_t)} (\tr
\Pi_t) .
\end{align*}
\end{proof}
Following lemma relates the trace of $X_{t_f}$ with the trace of $Y^*$ and $Y_{t_f}$.
\begin{lemma}\label{lem:final}
$\tr X_{t_f} \leq \frac{1}{(1-4\sqrt{\ve})} \cdot (\tr Y^*) \cdot
\ln (m/\tr Y_{t_f}) \enspace . $
\end{lemma}

\begin{proof}
Using Lemma~\ref{lem:thr} we have,

\begin{align*}
\frac{\tr Y_{t+1}}{\tr Y_{t}} & \leq 1 - \frac{ (1-4\sqrt{\ve}) \lambda_t \norm{\Phi^*(Y_t)} (\tr \Pi_t) }{\tr Y_t} \\
& \leq  \exp \left(- \frac{ (1-4\sqrt{\ve}) \lambda_t \norm{\Phi^*(Y_t)} (\tr \Pi_t) }{\tr Y_t} \right ) \quad (\text{since } \exp(-x) \geq 1 -x) \\
& \leq   \exp \left(- \frac{ (1-4\sqrt{\ve}) \lambda_t  \tr \Pi_t }{\tr Y^*}\right ) \quad (\text{from property of } Y^*) \\
& =   \exp \left(- \frac{ (1-4\sqrt{\ve})  \tr (X_{t+1} - X_t)}{\tr Y^*}\right)  .
\end{align*}
This implies,
\begin{align*}
 \tr Y_{t_f} & \leq (\tr Y_0) \exp\left(- \frac{ (1-4\sqrt{\ve}) \tr X_{t_f} }{\tr Y^*} \right) \\
 \Rightarrow \tr X_{t_f} & \leq \frac{(\tr Y^*) \ln (m/(\tr Y_{t_f}))}{(1-4\sqrt{\ve})} \quad (\text{since } \tr Y_0 = m) .
\end{align*}
\end{proof}
We can now finally bound the trace of $X^*$ in terms of the trace of $Y^*$.
\begin{theorem}
$X^*$ and $Y^*$ are feasible for the $P$ and $D$ respectively and
$$\tr X^*  \leq (1+5\sqrt{\ve}) \tr Y^* \enspace .$$
Therefore, since $\opt(P) = \opt(D)$,
$$\opt(D) = \opt(P) \leq \tr X^*  \leq (1+5\sqrt{\ve}) \tr Y^* \leq (1+5\sqrt{\ve}) \opt(D) = (1+5\sqrt{\ve}) \opt(P)  \enspace .$$

\end{theorem}
\begin{proof}
It is easily verified that $X^*$ and $Y^*$ are feasible for $P$ and $D$ respectively. From Lemma~\ref{lem:final} we have,
$$ \alpha \tr X^* =  \tr X_{t_f} \leq \frac{1}{1-4\sqrt{\ve}} \cdot (\tr Y^*) \cdot \ln (m/\tr Y_{t_f}) \enspace .$$
Since $Y_{t_f}  = \exp(- \Phi(X_{t_f}))$ we have
$$\tr Y_{t_f} = \tr \exp(- \Phi(X_{t_f})) \geq \norm{\exp(- \Phi(X_{t_f}))} = \exp (- \alpha) \enspace .$$
Using above two equations we have,
\begin{align*}
 \tr X^*  & \leq \frac{1}{1-4\sqrt{\ve}} \cdot (\tr Y^*) \cdot \frac{\ln (m/\tr Y_{t_f})}{\ln (1/\tr Y_{t_f})} \\
 & = \frac{1}{1-4\sqrt{\ve}} \cdot (\tr Y^*) \cdot \left( 1 + \frac{\ln m}{\ln (1/\tr Y_{t_f})}\right) \\
 & \leq \frac{1 + \ve}{1-4\sqrt{\ve}} \cdot (\tr Y^*) \quad (\text{since } \tr Y_{t_f} \leq \frac{1}{m^{1/\ve}}) \\
& \leq (1 + 5\sqrt{\ve}) \cdot \tr Y^*  \enspace .
\end{align*}
\end{proof}

\subsection{Time complexity}
In this section we are primarily interested in bounding the number of iterations of the algorithm, that is we will bound $k_f$ and also the number of iterations for any given $k$.
We claim, without going into further details, that the actions required by the algorithm in any given iteration can all be performed in time $\polylog(n) \cdot \polylog(m) \cdot \poly(\frac{1}{\ve})$
(since operations for Hermitian matrices like eigenspace decomposition, exponentiation, and other operations like sorting and binary search for a list of real numbers etc. can be all be performed in $\polylog$ time).

Let us first introduce some notation. Let $A$ be a Hermitian matrix and $l$
be a real number. Let
\begin{itemize}

\item $\Pi^A_l$ denote the projector onto the space spanned by the eigenvectors of $A$
with eigenvalues at least $l$. Let $\Pi^A$ be shorthand for $\Pi^A_1$.

\item $N_l(A)$ denote the sum of eigenvalues of $A$ at least $l$. Thus
$N_l(A)=\tr \Pi^A_l A$. Let $N(A)$ be shorthand for $N_1(A)$.

\item $\lambda_k(A)$ denote the k-th largest eigenvalue of $A$. 

\item $\lambda^{\downarrow}(A) \defeq (\lambda_1(A),\cdots,\lambda_n(A))$.

\item for any two vectors $u,v\in\mathcal{R}^n$ we say $u$ majorizes $v$, denoted $u\succeq v$, iff $\sum_{i=1}^ku_i=\sum_{i=1}^kv_i$ and for any $j \in [n]$ we have, 
$\sum_{i=1}^ju_i\geq\sum_{i=1}^jv_i$.

\end{itemize}

We will need the following facts.
\begin{fact}\label{thm:lownerorder}\cite{Bhatia96}
For $n \times n$ Hermitian matrices $A$ and $B$, $A \geq  B$ implies
$\lambda_i(A)\geq\lambda_i(B)$ for all $1\leq i\leq n$. Thus
$N_l(A)\geq N_l(B)$ for any real number $l$.
\end{fact}

\begin{fact}\label{thm:majorization}\cite{Bhatia96}
Let $A$ be an $n \times n$ Hermitian matrix and $P_1,\cdots,P_r$ be a
family of mutually orthogonal projections. Then
$\lambda^{\downarrow}(A)\succeq\lambda^{\downarrow}(\sum_iP_iAP_i).$
\end{fact}

\begin{fact}\cite{Jordan75}\label{lem:projectordecomposition}
For any two projectors $\Pi$ and $\Delta$, there exits an orthogonal
decomposition of the underlying vector space into one dimensional and two
dimensional subspaces that are invariant under both $\Pi$ and
$\Delta$. Moreover, inside each two-dimensional subspace, $\Pi$ and
$\Delta$ are rank-one projectors.
\end{fact}

\begin{lemma}\label{lem:numofphase}
Let $k_f$ be the final value of $k$. Then
$k_s-k_f=\mathcal{O}(\frac{\log m\log^2 n}{\ve^3})$.
\end{lemma}
\begin{proof}
Note that $\norm{\Phi^*(\I)} = \norm{\sum_{i=1}^m A_i} \leq m$, since
for each $i, \norm{A_i} \leq 1$. Hence
$$ k_s = \mathcal{O}((\log m) /\ve_0) \enspace .$$
Let $Y_{t_f-1} =\Diag(y_1, \ldots y_m)$.  We have (since $\tr A_i \geq \frac{1}{\gamma} \geq \frac{\ve^2}{m^2}$ for each $i$),
\begin{align*}
m (1+\ve_0)^{k_f+1} & \geq m \norm {\Phi^*(Y_{t_f-1})} \geq \tr {\Phi^*(Y_{t_f-1})} \\
& = \sum_{i=1}^m y_i \tr A_i \geq  \frac{\sum_{i=1}^m y_i}{\gamma} = \frac{\tr Y_{t_f-1}}{\gamma} \geq \frac{1}{m^{1/\ve}\gamma} \geq \frac{\ve^2}{m^{2+ 1/\ve}}\enspace .
\end{align*}
Hence $k_f \geq - \mathcal{O}(\frac{\log m }{\ve\ve_0})$. Therefore
$k_s - k_f = \mathcal{O}(\frac{\log m
}{\ve\ve_0})=\mathcal{O}(\frac{\log m\log^2 n}{\ve^3})$.
\end{proof}
\begin{theorem}\label{thm:timecomplexity}
For any fixed $k$, the number of iterations of the algorithm is at most 
$\O(\frac{\log^2 n}{\ve_1^9\ve})$.  Hence combined with Lemma \ref{lem:numofphase}, the total number of iterations of the algorithm is at most $\mathcal {O}(\frac{\log^{13}n\log m}{\ve^{13}}).$

\end{theorem}
\begin{proof}
Fix $k$. Assume that the Algorithm has reached step $3(d)$ for this fixed $k$ , $\frac{6 \log^2 n}{\ve_1^9\ve}$ times. As argued in the proof of Lemma~\ref{lem:thr}, whenever Algorithm reaches step $3(d)$, $\thr \geq k - \frac{3 \ln n}{\ve} $. Thus there exists a value $s$ between $k$ and $k - \frac{3 \ln n}{\ve}$ such that $\thr = s$ at least $\frac{ 2 \log n}{\ve_1^9}$ times.

From Lemma~\ref{lem:thr} we get that the sum of the eigenvalues above
$(1+\ve_0)^s$, is at most $n(1+\ve_1)(1+\ve_0)^s$ at the beginning of this phase. Whenever $\thr \neq s$ in this phase, using Fact \ref{thm:lownerorder},   we conclude that the eigenvalues of $\Phi^*(Y_t)$ above $(1+\ve_0)^s$ do not increase. 
Whenever $\thr = s$ in this phase, using Lemma
\ref{lem:mainlemma}, we conclude that the eigenvalues of $\Phi^*(Y_t)$ above $(1+\ve_0)^s$ reduce by a factor of $(1-\ve_1^9)$. This can be seen by letting $A$ in Lemma~\ref{lem:mainlemma} to be  $\frac{1 - \exp(-2\sqrt{\ve})}{(1+\ve_0)^s } \cdot \Phi^*( P^\geq_{\lambda_t} Y_t  P^\geq_{\lambda_t})$ and $B$ to be $\frac{1}{(1+\ve_0)^s} \Phi^*(Y^t) - A$. Now condition $3(d)(1.)$ of the Algorithm gives condition $(2)$ of Lemma~\ref{lem:mainlemma}. Condition $(1)$ of Lemma~\ref{lem:mainlemma} can also be seen to be satisfied (using Lemma~\ref{lem:normphi}) and condition $(4)$ of Lemma~\ref{lem:mainlemma} is false due to condition $3(c)$ of the Algorithm. This implies condition $(3)$  of Lemma~\ref{lem:mainlemma} must also be false which gives us the desired conclusion.

Therefore  the eigenvalues of $\Phi^*(Y_t)$ above $(1+\ve_0)^s$ (in particular above $(1+\ve_0)^k$) will
vanish before $\thr = s$, $\frac{ 2 \log n}{\ve_1^9}$ times. Hence $k$ must decrease before the Algorithm has reached step $3(d)$,  $\frac{6 \log^2 n}{\ve_1^9\ve}$ times.
\end{proof}
Following is a key lemma. It states that for two positive semidefinite matrices $A,B$, if $A$ has good weight in the large (above $1$) eigenvalues space of $A+B$ and if the sum of large (above $1$) eigenvalues of $B$ is pretty much the same as for $A+B$, then the sum of eigenvalues of $A+B$, slightly below $1$ should be a constant fraction larger than the sum above $1$. 
\begin{lemma}\label{lem:mainlemma}
Let $\ve'=\frac{\ve_0}{1+\ve_0}$. Let
$A,B$ be two $n \times n$ positive semidefinite matrices satisfying
\begin{eqnarray}
\|A+B\|\leq1+\ve_1 \text{ and }\ \|B\|\geq 1, \label{normab}\\
\tr \Pi^{A+B}A \geq\ve \tr \Pi^{A+B}(A+B), \text{ and }  \label{alarge}\\
\tr \Pi^BB \geq(1-\ve_1^9) \tr\Pi^{A+B}(A+B).\label{bdoesnotdecrease}
\end{eqnarray}
Then
\begin{equation}\label{belowthrlarge}
N_{1-\ve'}(A+B)>(1+\frac{2}{5}\ve)N(A+B).
\end{equation}
\end{lemma}
\begin{proof}
In order to prove this Lemma we will need to first show a few other Lemmas. 
By Fact \ref{lem:projectordecomposition}, $\Pi^B$ and $\Pi^{A+B}$
decompose the underlying space $V$ as follows,
$$V=\left(\bigoplus_{i=1}^kV_i \right) \bigoplus W .$$
Above for each $i \in [k]$, $V_i$ is either one-dimensional or two-dimensional
subspace, invariant for both $\Pi^B$ and $\Pi^{A+B}$ and inside $V_i$ at least one of $\Pi^B$ and $\Pi^{A+B}$
survives. $W$ is the subspace where both $\Pi^B$ and $\Pi^{A+B}$
vanish. We identify the subspace $V_i$ and the projector onto
itself. For any matrix $M$, define $M_i$  to be $V_i M V_i$. We can see that both the projectors $\Pi^B$ and $\Pi^{A+B}$ are decomposed
into the direct sum of one-dimensional projectors as follows.
$$\Pi^B=\bigoplus_{i=1}^k\Pi_i^B \quad \text{and}\quad \Pi^{A+B}=\bigoplus_{i=1}^k\Pi_i^{A+B}.$$
\begin{lemma}\label{lem:pi}
For any $i\in[k]$, $\Pi^{B_i}=\Pi^B_i$ and $\Pi^{A+B}_i=\Pi^{A_i+B_i}$.
That is, the eigenspace of $B_i$ with eigenvalues at least $1$, is
exactly the restriction of $\Pi^B$ to $V_i$ and similarly for $A_i+B_i$.
\end{lemma}
\begin{proof}
We prove $\Pi^{B_i}=\Pi^B_i$ and the other equality follows
similarly.  If $\dim V_i=1$, i.e. $V_i= \mathrm{span} \{|v\rangle\}$, then
either $\Pi^B|v\rangle=|v\rangle$ or $\Pi^B|v\rangle=0$. For the
first case, $\Pi^B_i=|v\rangle\langle v|$, and $B_i=\langle
v|B|v\rangle|v\rangle\langle v|$ and $\langle v|B|v\rangle\geq 1$,
which means $\Pi^{B_i}=|v\rangle\langle v|$. For the second case,
$\Pi^B_i=0$, $\langle v|B|v\rangle<1$, i.e. $\Pi^{B_i}=0$.

For the case $\dim V_i=2$, 
\begin{eqnarray*}
B_i&=&V_iBV_i=V_i(\Pi^BB\Pi^B+(\I-\Pi^B)B(\I-\Pi^B))V_i\\
&=&V_i(\bigoplus_j\Pi^B_j)B(\bigoplus_j\Pi^B_j)V_i+V_i((W\oplus \bigoplus_j(V_j-\Pi^B_j))B(W\oplus\bigoplus_j(V_j-\Pi^B_j)))V_i\\
&=&\Pi^B_iB\Pi^B_i+(V_i-\Pi^B_i)B(V_i-\Pi^B_i).
\end{eqnarray*}
Let $\Pi^B_i=|v_1\rangle\langle v_1|$ and
$V_i-\Pi^B_i=|v_0\rangle\langle v_0|$, then
\begin{equation}\label{eqn:specofB}
B_i=\langle v_1|B|v_1\rangle|v_1\rangle\langle v_1|+\langle
v_0|B|v_0\rangle|v_0\rangle\langle v_0|
\end{equation}
is the spectral decomposition of $B_i$. As
$\Pi^B|v_1\rangle=\Pi^B_i|v_1\rangle=|v_1\rangle$ and
$\Pi^B|v_0\rangle=\Pi^B_i|v_0\rangle=0$, we have $\langle
v_1|B|v_1\rangle\geq 1$ and $\langle v_0|B|v_0\rangle<1$, and hence
$\Pi^{B_i}=|v_1\rangle\langle v_1|.$
\end{proof}

\begin{lemma}\label{lem:pbbdecompose}
\begin{equation}\label{eqn:pbbdecompose}
\tr \Pi^BB =\sum_{i=1}^k\tr \Pi^{B_i}B_i ,
\end{equation}
\begin{equation}\label{eqn:pabbdecmpose}
\tr \Pi^{A+B}B =\sum_{i=1}^k \tr \Pi^{A_i+B_i}B_i , \text{ and }
\end{equation}
\begin{equation}\label{eqn:pababdecomposition}
\tr \Pi^{A+B}(A+B) =\sum_{i=1}^k\tr \Pi^{A_i+B_i}(A_i+B_i)
\end{equation}
Then using Eq.\eqref{alarge} and Eq.\eqref{bdoesnotdecrease} we get,
\begin{eqnarray}
\sum_{i=1}^k\tr \Pi^{A_i+B_i}B_i \leq(1-\ve)\sum_{i=1}^k\tr \Pi^{A_i + B_i}(A_i+B_i) .\label{eqn:aalarge}\\
\sum_{i=1}^k\tr \Pi^{B_i}B_i \geq(1-\ve_1^9)\sum_{i=1}^k\tr \Pi^{A_i+B_i}(A_i+B_i) .\label{eqn:bdoesnotdecrese}
\end{eqnarray}
\end{lemma}
\begin{proof}
We prove (\ref{eqn:pbbdecompose}) and (\ref{eqn:pabbdecmpose}) and
(\ref{eqn:pababdecomposition}) follow similarly.
\begin{eqnarray*}
\tr \Pi^BB &=&\sum_{i=1}^k\tr \Pi^B_iB =\sum_{i=1}^k\tr V_i\Pi^B_iV_i B = \sum_{i=1}^k\tr \Pi^B_iV_i BV_i \\
&=&\sum_{i=1}^k\tr \Pi^B_i B_i =\sum_{i=1}^k\tr \Pi^{B_i} B_i .
\end{eqnarray*}
\end{proof}

\bigskip

\noindent {\bf Remarks:}

\begin{enumerate}
\item In any one-dimensional subspace
$V_i=\mathrm{span}\{|v\rangle\}$ in the decomposition of $V$ as above, if $\Pi^{A+B}|v\rangle=0$, then
$\langle v|(A+B)|v\rangle<1$, which implies $\langle
v|B|v\rangle < 1$, that is  $\Pi^{B}|v\rangle=0$. But this
contradicts the fact that at least one of $\Pi^B$ and $\Pi^{A+B}$ does not vanish in $V_i$. Thus $\Pi^{A+B}$ never
vanishes in any of $V_i$. Therefore for all $i \in [k]$ we have $\tr \Pi^{A_i+B_i}(A_i+B_i) = \tr \Pi^{A+B}_i(A_i+B_i) \geq
1$.

\item From (\ref{normab}), for all $i \in [k]$, $\tr \Pi^{A_i+B_i}(A_i+B_i) \leq1+\ve_1$. Combined with
(\ref{eqn:pababdecomposition}), we have
%\begin{eqnarray}
$$k\leq N(A+B)\leq k(1+\ve_1).$$
%\label{eqn:upboundN1AB}
%\end{eqnarray}
\end{enumerate}

\begin{lemma}\label{lem:setij}
Let
$$I=\{i\in[k] : \tr \Pi^{A_i+B_i} B_i  \leq(1-\ve^2)\tr \Pi^{A_i+B_i}(A_i+B_i) \},$$
and
$$J=\{i\in[k] : \tr \Pi^{B_i}B_i \geq (1-\ve_1^8)\tr \Pi^{A_i+B_i}(A_i+B_i) \}.$$
Then
$$|I\cap J|>\frac{99}{100}\ve k.$$
\end{lemma}

\begin{proof}
From (\ref{eqn:aalarge}),
\begin{eqnarray*}
&&(1-\ve^2)\sum_{i\not\in
I}N(A_i+B_i)\leq(1-\ve)\sum_{i=1}^kN(A_i+B_i)\\
&\Rightarrow&(\ve-\ve^2)\sum_{i\not\in
I}N(A_i+B_i)\leq(1-\ve)\sum_{i\in
I}N(A_i+B_i)\\
&\Rightarrow&\ve(k-|I|)\leq(1+\ve_1)|I| \quad \mbox{(from Remarks 1. and 2.)}\\
&\Rightarrow&|I|\geq\frac{\ve}{1+\ve_1+\ve}k.
\end{eqnarray*}
From (\ref{eqn:bdoesnotdecrese}) (since for all $i\in [k], N(A_i + B_i) \geq N(B_i)$),
\begin{eqnarray*}
&&\sum_{i\in J}N(A_i+B_i)+(1-\ve_1^8)\sum_{i\not\in
J}N(A_i+B_i)\geq(1-\ve_1^9)\sum_{i=1}^kN(A_i+B_i)\\
&\Rightarrow&\ve_1\sum_{i\in J}N(A_i+B_i)\geq(1-\ve_1)\sum_{i\not\in
J}N(A_i+B_i)\\
&\Rightarrow&\ve_1(1+\ve_1)|J|\geq(1-\ve_1)(k-|J|)  \quad \mbox{(from Remarks 1. and 2.)}\\
&\Rightarrow&|J|\geq\frac{1-\ve_1}{1+\ve_1^2}k.
\end{eqnarray*}
Thus
$$|I\cap J|\geq \left(\frac{\ve}{1+\ve_1+\ve}+\frac{1-\ve_1}{1+\ve_1^2}-1 \right)k>\frac{99}{100}\ve k.$$
\end{proof}

\bigskip

\noindent\textbf{Remark:}

\begin{itemize}
\item[3.] Note that for any $i\in I\cap J$, $\dim V_i=2$. Otherwise, either
$\Pi^{A_i+B_i}=\Pi^{B_i}$ or $\Pi^{B_i}=0$ and neither of these  can
happen in $I\cap J$ (from definitions of $I$ and $J$).
\end{itemize}

The following lemma states that for each $i \in I \cap J$, the second eigenvalue of $A_i + B_i$ is close to $1$. Its proof involves some direct calculations and due to space constraint we move it to Appendix~\ref{sec:deferredproofs}.
\begin{lemma}\label{lem:2by2matrix}
Let $P$ and $Q$ be $2 \times 2$ positive semidefinite matrices satisfying
\begin{eqnarray}
\|Q\|\geq 1, \quad \|P+Q\|\leq1+\ve_1 , \quad  \lambda_2(P+Q)<1 ,
\label{eqn:2by2norm}\\
\tr \Pi^{P+Q}P \geq\ve^2\tr \Pi^{P+Q}(P+Q)  \quad  \text{and }
\label{eqn:2by2alarge}\\
\tr \Pi^QQ \geq(1-\ve_1^8)\tr \Pi^{P+Q}(P+Q)  \enspace .\label{eqn:2by2bdoesnotdecrese}
\end{eqnarray}
Then $\lambda_2(P+Q)>1-\frac{1}{9}\ve_1^3.$
\end{lemma}

\bigskip

\noindent We can finally prove  Lemma \ref{lem:mainlemma}. By Fact
\ref{thm:lownerorder},
$\lambda^{\downarrow}(A+B)\succeq\lambda^{\downarrow}(\sum_i{A_i+B_i})$.
Let $j_1=\max\{j: \lambda_j(A+B)\geq 1\}$,
$j_2=\max\{j:\lambda_j(\sum_i(A_i+B_i))\geq 1\},$  and
$j_0=j_1+\frac{99}{100}\ve k$. Then
$$\sum_{j\leq j_0}\lambda_j(A+B)\geq\sum_{j\leq j_0}\lambda_j\(\sum_i(A_i+B_i)\).$$

According to the decomposition in Fact
\ref{lem:projectordecomposition}, Lemma \ref{lem:pbbdecompose} and
the remarks below it, $j_1=j_2=k$ and
$$\sum_{j\leq j_1}\lambda_j(A+B)=\tr \Pi^{A+B}(A+B), \quad \text{and}$$
$$\sum_{j\leq
j_2}\lambda_j\(\sum_i(A_i+B_i)\)=\sum_i\tr \Pi^{A_i+B_i}(A_i+B_i).$$
The RHS of both the equations are equal by Lemma \ref{lem:pbbdecompose}. Therefore,
$$\sum_{k<j\leq j_0}\lambda_j(A+B)\geq\sum_{k<j\leq j_0}\lambda_j\(\sum_i(A_i+B_i)\).$$
By Lemma \ref{lem:setij}
and Lemma \ref{lem:2by2matrix},
$$\sum_{k<j\leq j_0}\lambda_j\(\sum_i(A_i+B_i)\)\geq\frac{99}{100}\ve k\(1-\frac{1}{9}\ve_1^3\).$$
Let $x=N_{1-\ve'}(A+B)-N(A+B)$, then
$$\sum_{k<j\leq j_0}\lambda_j(A+B)\leq x+\(\frac{99}{100}\ve k-x\)(1-\ve').$$
Therefore from previous three inequalities,
$$\frac{99}{100}\ve k\(1-\frac{1}{9}\ve_1^3\)\leq x+\(\frac{99}{100}\ve k-x\)(1-\ve'),$$
which implies
$$x\geq\frac{99}{100}\ve k\(1-\frac{\ve_1^3}{9\ve'}\).$$
Note that $\ve_1^3\ll\ve'$, therefore from Remark 2.,
\begin{align*}
N_{1-\ve'}(A+B)&\geq k+\frac{99}{100}\ve k\(1-\frac{\ve_1^3}{9\ve'}\)>\(1+\frac{1}{2}\ve\)k\\
&>\(1+\frac{2}{5}\ve\)(1+\ve_1)k\geq\(1+\frac{2}{5}\ve\)N(A+B).
\end{align*}
\end{proof}

\subsubsection*{Acknowledgement}
Penghui Yao would like to thank Attila Pereszl$\acute{e}$nyi and Huangjun Zhu for helpful discussions.

\bibliographystyle{plain}

\bibliography{fastspd}

\appendix

\section{Transforming to special form}
\label{sec:transform}

Let us consider an instance of a positive semidefinite program as follows.

\begin{center}
  \begin{minipage}{2in}\vspace{-10mm}
    \centerline{\underline{Primal problem P}}\vspace{-4mm}
    \begin{align*}
      \text{minimize:}\quad & \tr C X\\
      \text{subject to:}\quad & \forall i\in  [m] : \tr A_i X \geq b_i,\\
      & X \geq 0.
    \end{align*}
  \end{minipage}
  \hspace*{25mm}
  \begin{minipage}{2in}
    \centerline{\underline{Dual problem D}}\vspace{-4mm}
    \begin{align*}
      \text{maximize:}\quad & \sum_{i=1}^m b_i y_i \\
      \text{subject to:}\quad & \sum_{i=1}^m y_i \cdot A_i  \leq C,\\
      & \forall i \in [m] : y_i \geq 0.
    \end{align*}
  \end{minipage}
\end{center}
Above $C, A_1, \ldots, A_m$ are $n \times n$ positive semidefinite matrices and $b_1, \ldots, b_m$ are non-negative.  Let us assume that conditions for strong duality are satisfied and optimum value for ${P}$, denoted $\opt({P})$, equals the optimum value for ${D}$, denoted $\opt({D})$. Assume w.l.o.g $m \geq n$ (by repeating the first constrain in $P$ if necessary).

We show how to transform the primal problem to the special form and a similar
transformation can be applied to dual problem. First observe that if for some $i$,
$b_i=0$, the corresponding constraint in primal problem is trivial and can be removed.
Similarly if for some $i$, the support of $A_i$ is not contained in the support of $C$, then
$y_i$ must be $0$ and can be removed. Therefore we can assume w.l.o.g. that for all $i, b_i > 0$ and the support of $A_i$ is contained in the support of $C$. Hence w.l.o.g we can take the support of $C$ as the
whole space, in other words, $C$ is invertible. For all $i \in [m],$ define
$A_i' \defeq \frac{C^{-1/2}A_iC^{-1/2}}{b_i}$. Consider the normalized Primal problem.

\begin{center}
  \begin{minipage}{2in}
    \centerline{\underline{Normalized Primal problem P'}}\vspace{-4mm}
    \begin{align*}
      \text{minimize:}\quad & \tr X'\\
      \text{subject to:}\quad & \forall i\in  [m] : \tr A_i' X' \geq 1,\\
      & X' \geq 0.
    \end{align*}
  \end{minipage}
\end{center}

\begin{claim} If $X$ is a feasible solution to $P$, then $C^{1/2}XC^{1/2}$
is a feasible solution to $P'$ with the same objective value. Similarly if $X'$ is a feasible solution to $P'$, then
$C^{-1/2}X'C^{-1/2}$ is a feasible solution to $P$ with the same
objective value. Hence $\opt(P) = \opt(P')$.
\end{claim}

\begin{proof}
Easily verified.
\end{proof}
The next step to transforming the problem is to limit the range of
eigenvalues of $A_i'$s. Let $\beta=\min_i{\|A_i'\|}$.

\begin{claim}\label{claim:opt}
$\frac{1}{\beta}\leq \opt(P')\leq\frac{m}{\beta}.$
\end{claim}
\begin{proof}
Note that $\frac{1}{\beta} \I$ is a feasible solution for $P'$. This implies $\opt(P') \leq \frac{n}{\beta} \leq \frac{m}{\beta}$. Let $X'$ be an optimal feasible solution for $P'$.  Let $j$ be such that $\|A'_j\| = \beta$. Then $\beta \tr X' \geq \tr A'_j X' \geq 1$, hence $\frac{1}{\beta}\leq \opt(P')$.
\end{proof}
Let $A_i'=\sum_{j=1}^n a_{ij}'|v_{ij}\rangle\langle v_{ij}|$ be the
spectral decomposition of $A_i'$. Define for all $i \in [m]$ and $j \in [n]$,
\begin{equation}\label{eqn:aij}
a_{ij}^{''}\defeq\begin{cases}\frac{\beta m}{\ve} &\text{if
$a_{ij}'>\frac{\beta m}{\ve}$,}\\ 0 &\text{if
$a_{ij}'<\frac{\ve\beta}{m},$}\\ a_{ij}'
&\text{otherwise.}\end{cases}
\end{equation}
Define $A_i^{''}=\sum_{j=1}^n a_{ij}^{''}|v_{ij}\rangle\langle v_{ij}|.$
Consider the transformed Primal problem $P^{''}$.

\begin{center}
  \begin{minipage}{2in}
  \centerline{\underline{Transformed Primal problem $P^{''}$}}\vspace{-4mm}
       \begin{align*}
      \text{minimize:}\quad & \tr X^{''}\\
      \text{subject to:}\quad & \forall i\in  [m] : \tr A_i^{''} X^{''} \geq 1, \\
      & X^{''} \geq 0.
    \end{align*}
  \end{minipage}
\end{center}

\begin{lemma}\label{lem:specialform}
\begin{enumerate}
\item Any feasible solution to $P^{''}$ is also a feasible solution to
$P'$.

\item $\opt(P') \leq \opt(P^{''}) \leq \opt(P')(1+\ve)$.
\end{enumerate}

\end{lemma}

\begin{proof}
\begin{enumerate}
\item Follows immediately from the fact that $A_i^{''}\leq A_i'$.

\item First inequality follows from 1. Let $X'$ be an optimal solution to $P'$ and let $\tau=\tr(X')$. Let
$X^{''}=X'+\frac{\ve\tau}{m}\I$. Then, since $m \geq n$, $\tr X^{''}\leq(1+\ve)\tr X'$.
Thus it suffices to show that $X^{''}$ is feasible to $P^{''}$.

Fix $i \in [m]$. Assume that there exists $j \in [n]$ such that $a_{ij}'\geq\frac{\beta m}{\ve}$. Then, from Claim~\ref{claim:opt}
$$\tr A_i^{''}X_i^{''} \geq\tr \frac{\beta m}{\ve}|v_{ij}\rangle\langle v_{ij}|\cdot \frac{\ve\tau}{m}\I =\beta\tau\geq 1.$$

Now assume that for all $j \in [n]$, $a_{ij}\leq\frac{\beta m}{\ve}$. By
(\ref{eqn:aij}) and definition of $\beta$, $\|A_i''\| = \|A_i'\| \geq \beta$ and $A_i'' \geq A_i' - \frac{\ve \beta}{m} \I$. Therefore
\begin{align*}
\tr A_i^{''}X_i^{''} &\geq \tr A_i''X' +\beta\frac{\ve\tau}{m} \\
& \geq \tr A_i'X' +\beta\frac{\ve\tau}{m}-\tr \frac{\ve\beta}{m}X' =\tr A_i'X' \geq 1.
\end{align*}
\end{enumerate}
\end{proof}

Note that for all $i \in [m]$, the ratio between the largest eigenvalue and
the smallest nonzero eigenvalue of $A_i^{''}$ is at most $\frac{m^2}{\ve^2}=\gamma$.

Finally, we get the special form Primal problem $\hat{P}$ as follows. Let $t = \max_{i\in [m]}\|A^{''}_i\|$ and for all $i \in [m]$ define $\hat{A}_i \defeq \frac{A_i^{''}}{t}$. Consider,

\begin{center}
  \begin{minipage}{2in}
  \centerline{\underline{Special form Primal problem $\hat{P}$}}\vspace{-4mm}
       \begin{align*}
      \text{minimize:}\quad & \tr \hat{X}\\
      \text{subject to:}\quad & \forall i\in  [m] : \tr \hat{A}_i \hat{X} \geq 1, \\
      & \hat{X} \geq 0.
    \end{align*}
  \end{minipage}
\end{center}
It is easily seen that there is a one-to-one correspondence between the
feasible solutions to $P^{''}$ and $\hat{P}$ and $\opt(\hat{P}) = t \cdot \opt(P'')$. Therefore $\hat{P}$ satisfies all the properties that we want and
cumulating all we have shown above, we get following conclusion.
\begin{lemma}\label{lem:specialformfinal}
Let $\hat{X}$ be a feasible solution to $\hat{P}$ such that $\tr \hat{X} \leq (1+\ve) \opt(\hat{P})$. A feasible solution $X$
to $P$ can be derived from $\hat{X}$ such that $\tr X \leq (1+\ve)^2 \opt(P)$.
\end{lemma}
Furthermore we claim, without giving further details, that $X$ can be obtained from $\hat{X}$ in time $\polylog(n)\cdot \polylog(m)$.

\section{Deferred Proofs}
\label{sec:deferredproofs}
\begin{proofof}{Lemma~\ref{lem:2by2matrix}}
Let $\eta$ be the maximum real number such that
$P-\eta(\I-\Pi^{P+Q}) \geq 0$. Set $P_1=P-\eta(\I-\Pi^{P+Q})$ and
$Q_1=Q+\eta(\I-\Pi^{P+Q})$. $P_1,Q_1$ satisfy all the conditions in
this Lemma and $P_1$ is a rank one matrix. Furthermore, set
$P_2=P_1/\|Q_1\|$ and $Q_2=Q_1/\|Q_1\|$. Again all the conditions in this Lemma
are still satisfied by $P_2, Q_2$ since $\Pi^{Q_2}=\Pi^{Q_1}=\Pi^Q$
and $\Pi^{P_2+Q_2}=\Pi^{P_1+Q_1}=\Pi^{P+Q}$. As
$\lambda_2(P_2+Q_2)\leq\lambda_2(P_1+Q_1) = \lambda_2(P+Q),$ it suffices to prove that $\lambda_2(P_2+Q_2)>1-\frac{1}{9}\ve_1^3$. Consider $P_2, Q_2$ in the diagonal bases of $Q_2$.
\[P_2=
   \left(\begin{array}{cc}
     |r|\cos^2\theta & r\sin\theta\cos\theta\\
     r^*\sin\theta\cos\theta & |r|\sin^2\theta
     \end{array}\right),
Q_2=
   \left(\begin{array}{cc}
     1 & 0\\
     0 & b
     \end{array}\right).
\]
where $r \in \complex$ and $0\leq b<1$. Set $\lambda=\|P_2+Q_2\|$. Eq.~\eqref{eqn:2by2bdoesnotdecrese} implies that
\begin{equation}\label{eqn:upbforlambda}
\lambda\leq\frac{1}{1-\ve_1^8}<1+2\ve_1^8.
\end{equation}
Since
\begin{eqnarray*}
\tr \Pi^{Q_2}P_2 &=&\tr \Pi^{Q_2}(P_2+Q_2) -\tr \Pi^{Q_2}Q_2 \leq
\tr \Pi^{P_2+Q_2}(P_2+Q_2) -\tr \Pi^{Q_2}Q_2 \\
&\leq&\ve_1^8\tr \Pi^{P_2+Q_2}(P_2+Q_2) =\ve_1^8\lambda<2\ve_1^8,
\end{eqnarray*}
we have,
\begin{equation}\label{eqn:rcos}
|r|\cos^2\theta<2\ve_1^8.
\end{equation}
Observe that,
$$|v\rangle=    \frac{1}{\sqrt{1+\left(\frac{|r|\sin\theta\cos\theta}{\lambda-b-|r|\sin^2\theta}\right)^2}} \ 
 \left(\begin{array}{c}
     1 \\
     \frac{r^*\sin\theta\cos\theta}{\lambda-b-|r|\sin^2\theta}
     \end{array}\right),$$
is the eigenvector of $P_2 + Q_2$ with eigenvalue $\lambda$. Hence $\Pi^{P_2+Q_2}=|v\rangle\langle v|$. Note that~$\lambda >
b+|r|\sin^2\theta$, because $\lambda_2(P_2+Q_2)=1+|r|+b-\lambda<1.$ Consider 
%(below $\mathbb {R}(\cdot)$ represents real part),
\begin{eqnarray*}
\tr(\Pi^{P_2+Q_2}P_2)&=&\langle v|P_2|v\rangle\\
%&=&\frac{|r|\cos^2\theta+ \frac{2 {\mathbb {R}}(r^2)\sin^2\theta\cos^2\theta}{\lambda-b-|r|\sin^2\theta}+ \frac{|r|^3\sin^4\theta\cos^2\theta}{(\lambda-b-  |r|\sin^2\theta)^2}} 
%{1+\frac{|r|^2\sin^2\theta\cos^2\theta}{(\lambda-b-|r|\sin^2\theta)^2}}\\
&=&\frac{|r|\cos^2\theta+ \frac{2 |r|^2\sin^2\theta\cos^2\theta}{\lambda-b-|r|\sin^2\theta}+ \frac{|r|^3\sin^4\theta\cos^2\theta}{(\lambda-b-  |r|\sin^2\theta)^2}} 
{1+\frac{|r|^2\sin^2\theta\cos^2\theta}{(\lambda-b-|r|\sin^2\theta)^2}}\\
&=&\frac{|r|\cos^2\theta(\lambda-b-|r|\sin^2\theta)^2+2|r|^2\sin^2\theta\cos^2\theta(\lambda-b-|r|\sin^2\theta)+|r|^3\sin^4\theta\cos^2\theta}{(\lambda-b-|r|\sin^2\theta)^2+|r|^2\sin^2\theta\cos^2\theta}\\
&=&\frac{|r|(\lambda-b)^2\cos^2\theta}{(\lambda-b-|r|\sin^2\theta)^2+|r|^2\sin^2\theta\cos^2\theta}\\
&\leq&\frac{|r|\cos^2\theta}{(1-\frac{|r|\sin^2\theta}{\lambda-b})^2} \\
&<&\frac{2\ve_1^8}{(1-\frac{|r|\sin^2\theta}{\lambda-b})^2}.
\end{eqnarray*}
Combining with (\ref{eqn:2by2alarge}), we obtain
\begin{align*}
\lefteqn{2\ve_1^8\geq\ve^2(1-\frac{|r|\sin^2\theta}{\lambda-b})^2} \\
& \Rightarrow(1-\frac{|r|\sin^2\theta}{\lambda-b})^2<\frac{\ve_1^6}{100} \\
& \Rightarrow |r|\sin^2\theta>(1-\frac{1}{10}\ve_1^3)(\lambda-b)\\
&\Rightarrow |r|\sin^2\theta+(1-\frac{1}{10}\ve_1^3)b>(1-\frac{1}{10}\ve_1^3)\lambda \\
& \Rightarrow |r|+b>(1-\frac{1}{10}\ve_1^3)\lambda>1-\frac{1}{10}\ve_1^3 \enspace .
\end{align*}
Hence
\begin{align*}
\lambda_2(P_2+Q_2) & =\tr(P_2+Q_2)-\lambda=1+|r|+b-\lambda \\
& >2-\frac{1}{10}\ve_1^3-(1+2\ve_1^8)>1-\frac{1}{9}\ve_1^3 \ .
\end{align*}
\end{proofof}

\end{document}